\documentclass[twocolumn]{IEEEtran} 
\makeatletter
\let\NAT@parse\undefined
\makeatother

\usepackage{amsmath,amssymb,amsfonts,cuted}
\usepackage[final]{graphicx}
\usepackage{psfrag}
\usepackage{epsfig}
\usepackage[numbers,sort&compress]{natbib}

\usepackage{array,booktabs}

\usepackage{amsmath,amssymb,amsfonts,amsthm}
\usepackage[final]{graphicx}

\usepackage[nolist,nohyperlinks]{acronym}

\usepackage{ucs} 
\usepackage[utf8x]{inputenc}
\usepackage[usenames,dvipsnames]{xcolor}
\usepackage{tikz}
\usepackage{tkz-tab}
\usepackage{multirow}
\usepackage{latexsym}
\usepackage{amssymb}
\usepackage{amsmath}
\usepackage{mathrsfs}
\usepackage{amsthm}

\newtheorem{theorem}{Theorem}
\newtheorem{lemma}[theorem]{Lemma}

\begin{acronym}
\acro{SNR}{signal-to-noise ratio}
\acrodefplural{SNR}[SNRs]{signal-to-noise ratios}
\acro{AWGN}{additive white Gaussian noise}
\acro{CSI}{channel state information}
\acro{PDF}{probability density function}
\acrodefplural{PDF}[PDFs]{probability density functions}
\acro{CDF}{cumulative distribution function}
\acro{GMGF}{generalized moment generating function}
\acro{MGF}{moment generating function}
\acrodefplural{CDF}[CDFs]{cumulative distribution functions}
\acro{RV}{random variable}
\acro{LoS}{line-of-sight}
\acro{MRT}{maximal ratio transmission}
\acro{MRC}{maximal ratio combining}
\acro{dR}{double-Rayleigh}
\acro{CLT}{central limit theorem}
\acro{i.i.d.}{independent and identically distributed}
\acro{fdRLoS}{fluctuating double-Rayleigh with line-of-sight}
\acro{fSOSF}{fluctuating second order scattering fading}
\acro{dRLoS}{double-Rayleigh with line-of-sight}
\acro{SOSF}{second order scattering fading}
\acro{OP}{outage probability}
\acro{MC}{Monte Carlo}
\end{acronym}

\begin{document}

\title{\huge{The Second Order Scattering Fading Model\\ with Fluctuating Line-of-Sight}}
\author{Jes\'us L\'opez-Fern\'andez, Gonzalo J. Anaya-L\'opez, and F.~Javier L\'opez-Mart\'inez

\thanks{Manuscript received MONTH xx, YEAR; revised XXX. The review of this paper was coordinated by XXXX. This work was funded in part by Junta de Andaluc\'ia, the European Union and the European Fund for Regional Development FEDER through grants P18-RT-3175 and EMERGIA20-00297, in part by MCIN/AEI/10.13039/501100011033 through grant PID2020-118139RB-I00, and in part by Universidad de M\'alaga (UMA20-FEDERJA-002). (\textit{Corresp: Jes\'us L\'opez-Fern\'andez}). This work has been submitted to the IEEE for possible publication. Copyright may be transferred without notice, after which this version may no longer be accessible.}
\thanks{The authors are with the Communications and Signal Processing Lab, Telecommunication Research Institute (TELMA), Universidad de M\'alaga, M\'alaga, 29010, (Spain). F.~J.~L\'opez-Mart\'inez is also with the Dept. Signal Theory, Networking and Communications, University of Granada, 18071, Granada (Spain). (E-mails: $\rm jlf@ic.uma.es, gjal@ic.uma.es, fjlm@ugr.es$)}



}
\maketitle

\begin{abstract}
We present a generalization of the notoriously unwieldy second-order scattering fading model, which is helpful to alleviate its mathematical complexity while providing an additional degree of freedom. This is accomplished by allowing its dominant specular component associated to line-of-sight propagation to randomly fluctuate. The statistical characterization of the newly proposed model is carried out, providing {closed-form} expressions for its probability and cumulative distribution functions, as well as for its generalized Laplace-domain statistics and raw moments. We exemplify how performance analysis can be done in this scenario, and discuss the role of the fading model parameters on system performance.
\end{abstract}


\section{Introduction}
\IEEEPARstart{D}{ouble-scattering} fading conditions often appear in radio propagation environments like vehicular communications \cite{Bithas2016,Ai2018,Bithas2018}, unmanned aerial vehicle-enabled communications \cite{Bithas2020}, backscattering systems \cite{Devineni2019,Unai2020}, indoor mobile \cite{Vinogradov2015}, or land mobile satellite channels \cite{Nikolaidis2018}. A family of multiple-order scattering fading channels was originally defined by Andersen \cite{Andersen2002} and formalized by Salo \cite{Salo2006}, so that a finite number of increasing-order scattering terms is considered. Even for its simplest formulation that consists of a \ac{LoS} component plus Rayleigh \textit{and} \ac{dR} diffuse components, referred to as \ac{SOSF} model, its mathematical complexity and potential numerical instability have limited the applicabilty of this otherwise useful fading model. Only recently, an alternative formulation for the \ac{SOSF} model proposed in \cite{Lopez2018} provided a reasonably simpler approach that fully avoided the original numerical issues suffered by the model, and the \ac{MGF} of the \ac{SOSF} model was derived for the first time.

Several state-of-the-art fading models incorporate the ability to model random amplitude fluctuations of dominant specular waves associated to \ac{LoS} propagation. Relevant examples include popular fading models like Rician shadowed \cite{Abdi2003} and its generalizations \cite{Paris2014,Romero2022}. Recently, a \ac{fdRLoS} fading model was formulated as a combination of a randomly fluctuating \ac{LoS} component plus a \ac{dR} diffuse one, {although the first-order Rayleigh-like component also present in the original \ac{SOSF} model is neglected}. In this work, we define a natural generalization of the \ac{SOSF} model to incorporate random fluctuations on its \ac{LoS} component, for which the moniker \ac{fSOSF} is proposed. The newly proposed model is able to capture the same propagation conditions as the baseline \ac{SOSF} model, and allows to tune the amount of fluctuation suffered by the dominant component through one additional parameter. Interestingly, the addition of a new parameter for this model does not penalize its mathematical tractability, and the resulting expressions for its chief statistics have the same functional form (even simpler in some cases) {as those of the} original \ac{SOSF} model. The applicability of the \ac{fSOSF} model for performance analysis purposes is also exemplified through several illustrative examples.

\textit{Notation}: $\mathbb{E}\{X\}$ and $|X|$ denote the statistical average and the modulus of the complex \ac{RV} $X$ respectively. The \ac{RV} $X$ conditioned to $Y$ will be denoted as $X|Y$. The symbol $\sim$ reads as \emph{statistically distributed as}. The symbol $\stackrel{d}{=}$ reads as \emph{equal in distribution}. A circularly symmetric normal \ac{RV} $X$ with mean $\mu$ and variance $\Omega$ is denoted as $X\sim \mathcal{N}_c(\mu,\Omega)$.
%

\section{Physical model}
\label{Sec:The system model}
Based on the original formulation of the \ac{SOSF} model introduced by Andersen \cite{Andersen2002} and Salo \cite{Salo2006}, let us consider the following definition for the received signal $S$ as
\begin {equation}
S=\omega_0 \sqrt{\xi} e^{j\phi}+\omega_1 G_1+\omega_2 G_2 G_3,
\label{Eq:Modelo_fSOSF}
\end{equation}
where $\omega_0 e^{j\phi}$ is the dominant specular component classically associated to \ac{LoS} propagation, with $\omega_0$ being a constant value and $\phi$ a \ac{RV} uniformly distributed in $[0,2\pi)$. The \acp{RV} $G_1$, $G_2$ and $G_3$ are distributed as independent zero-mean, unit-variance complex normal variables, i.e., $G_i\sim\mathcal{N}_c(0,1)$ for $i=1,2,3$. {The constant parameters $\omega_0$, $\omega_1$ and $\omega_2$ act as scale weights for the \ac{LoS}, Rayleigh and \ac{dR} components, respectively}. Now, the key novelty of the model in \eqref{Eq:Modelo_fSOSF} lies on its ability to incorporate random fluctuations into the \ac{LoS} similarly to state-of-the-art fading models in the literature \cite{Abdi2003,Paris2014} through $\xi$, which is a Gamma distributed \ac{RV} with unit power and real positive shape parameter $m$, with \ac{PDF}:
\begin{equation}
f_{\xi}(u)=\frac{m^mu^{m-1}}{\Gamma(m)}e^{-m u},
\end{equation}
where $\Gamma(\cdot)$ is the gamma function. The severity of \ac{LoS} fluctuations is captured through the parameter $m$, being fading severity inversely proportional to this shape parameter. In the limit case of $m\rightarrow\infty$, $\xi$ degenerates to a deterministic unitary value and the \ac{LoS} fluctuation vanishes, thus collapsing into the original \ac{SOSF} distribution. 
 
Besides $m$, the \ac{fSOSF} model is completely defined by the constants $\omega_0$, $\omega_1$ and $\omega_2$. Typically, an alternative set of parameters is used in the literature for the baseline \ac{SOSF} model, i.e. ($\alpha, \beta$), defined as
\begin {equation}
\alpha=\frac{\omega_2^2}{\omega_0^2+\omega_1^2+\omega_2^2},\;\;\;\;\;\beta=\frac{\omega_0^2}{\omega_0^2+\omega_1^2+\omega_2^2}.
\label{Eq:alpha_beta}
\end{equation}
Assuming a normalized channel (i.e., $\mathbb{E}\{|S|^2\}=1$) so that $\omega_0^2+\omega_1^2+\omega_2^2=1$, the parameters $(\alpha,\beta)$ are constrained to the triangle $\alpha \geq 0$, $\beta \geq 0$ and $\alpha+\beta \leq 1$.

\section{Statistical Characterization}
\label{Sec:3}
Let us define the instantaneous \ac{SNR} $\gamma=\overline\gamma|S|^2$, where $\overline\gamma$ is the average \ac{SNR}. The model in \eqref{Eq:Modelo_fSOSF} reduces to the \ac{SOSF} one \cite{Salo2006} when conditioning to $\xi$. However, it is possible to find an alternative pathway to connect this model with a different underlying model in the literature, so that its mathematical formulation is simplified. 

According to \cite{Lopez2018}, the \ac{SOSF} model can be seen as a Rician one when conditioning to $x=|G_3|^2$. Hence, this observation can be leveraged to formulate the \ac{fSOSF} model in terms of an underlying Rician shadowed one, as shown in the sequel. For the \ac{RV} $\gamma$ we can express:
\begin {equation}
\gamma=\overline\gamma|\omega_0 {\color{black}\sqrt{\xi}}e^{j\phi}+\omega_1 G_1+\omega_2 G_2 G_3|^2.
\label{Eq:gamma}
\end{equation}
Since $G_3$ is a complex Gaussian RV, we reformulate $G_3=|G_3|e^{j\Psi}$, where $\Psi$ is uniformly distributed in $[0,2\pi)$. Because $G_2$ is a circularly-symmetric \ac{RV}, $G_2$ and $G_2e^{j\Psi}$ are equivalent in distribution, so that the following equivalence holds for $\gamma$ 
\begin {equation}
\gamma\stackrel{d}{=}\overline\gamma|\omega_0 {\color{black}\sqrt{\xi}}e^{j\phi}+\omega_1 G_1+\omega_2 G_2| G_3||^2.
\label{Eq:gamma2}
\end{equation}
Conditioning on $x=|G_3|^2$, define the conditioned \ac{RV} $\gamma_x$ as
\begin {equation}
\gamma_x\triangleq\overline\gamma|\omega_0 {\color{black}\sqrt{\xi}}e^{j\phi}+\omega_1 G_1+\omega_2 \sqrt{x} G_2|^2.
\label{Eq:gamma3}
\end{equation}
where the two last terms correspond to the sum of two RVs distributed as $\mathcal{N}_c(0;\omega_1^2)$ and $\mathcal{N}_c(0;\omega_2^2x)$, respectively. This is equivalent to one single RV distributed as $\mathcal{N}_c(0;\omega_1^2+\omega_2^2x)$. With all these considerations,$\gamma_x$ is distributed according to a squared Rician shadowed \ac{RV} \cite{Abdi2003} with parameters $m$ and
\begin{align}
\overline{\gamma}_x&=\frac{\omega_0^2}{\omega_1^2+x \omega_2^2}=\overline{\gamma}(1-\alpha(1-x)), \label{Eq_gamma_x}\\
K_x&=\omega_0^2+\omega_1^2+x \omega_2^2=\frac{\beta}{1-\beta-\alpha(1-x)} \label{Eq_K_x}.
\end{align}
We note that these parameter definitions include as special case the model in \cite{Lopez2022}, when $\omega_1^2=0$. In the following set of Lemmas, the main statistics of the \ac{fSOSF} distribution are introduced for the first time in the literature; these include the \ac{PDF}, \ac{CDF}, \ac{GMGF} and the moments.

\begin{lemma}\label{lemma1}
Let $\gamma$ be an \ac{fSOSF}-distributed \ac{RV} with shape parameters $\{\alpha,\beta,m\}$, i.e., $\gamma\sim\mathcal{F}_{\rm SOSF}\left(\alpha,\beta,m;\overline\gamma\right)$. Then, the \ac{PDF} of $\gamma$ is given by
\begin{align}\label{eqpdf1}
f_\gamma(\gamma)=\int_{0}^{\infty}&\tfrac{m^m(1+K_x)}{(m+K_x)^m\overline\gamma_x}e^{-\tfrac{1+K_x}{\overline\gamma_x}\gamma-x}\times\nonumber\\&{}_1F_{1}\left(m;1;\tfrac{K_x(1+K_x)}{K_x+m}\tfrac{\gamma}{\overline\gamma_x}\right)dx,
 \end{align}
\begin{align} \nonumber
f_{\gamma}(\gamma)=&\sum_{j=0}^{m-1} \tbinom{m-1}{j} \frac{\gamma^{m-j-1} e^{\tfrac{m(1-\alpha-\beta)+\beta)}{m\alpha}}\left( \tfrac{\beta}{m}\right)^{m-j-1} }{(\overline{\gamma})^{m-j}(m-j-1)! \alpha^{2m-j-1}}\times \\ 
&\sum_{r=0}^{j} \binom{j}{r} \left( \tfrac{-\beta}{m}\right)^{j-r} \alpha^{r} \times \nonumber \\
&\Gamma\left(r-2m+j+2, \tfrac{m(1-\alpha-\beta)+\beta)}{m\alpha}, \tfrac{\gamma}{\alpha\overline{\gamma}} \right),
\label{Eq_PDF}
\end{align}
for $m\in\mathbb{R}^+$ and $m\in\mathbb{Z}^+$, respectively, and where $_1F_{1}\left(\cdot;\cdot;\cdot\right)$ and $\Gamma(a,z,b)=\int_{z}^{\infty}t^{a-1} e^{-t}e^{\tfrac{-b}{t}}dt$ are Kummer's hypergeometric funcion, and a generalization of the incomplete gamma function defined in \cite{CHAUDHRY199499}, respectively.
\end{lemma}

\begin{proof}
See Appendix \ref{ap1}.
\end{proof}

\begin{lemma}\label{lemma2}
Let $\gamma\sim\mathcal{F}_{\rm SOSF}\left(\alpha,\beta,m;\overline\gamma\right)$. Then, the \ac{CDF} of $\gamma$ is given by
\begin{align} \nonumber
\label{CDF_SOSFShadowed}
F_{\gamma}(\gamma)=&1-e^{\tfrac{m(1-\alpha-\beta)+\beta)}{m\alpha}} \sum_{j=0}^{m-1}\sum_{r=0}^{m-j-1}\sum_{q=0}^{j}\tbinom{m-1}{j}\tbinom{j}{q}  \times \\
&\tfrac{(-1)^{j-q} \alpha^{q-r-m+1}}{r!} \left(\tfrac{\gamma}{\overline \gamma}\right)^r \left( \tfrac{\beta}{m}\right)^{m-q-1}\times\\ \nonumber
&\Gamma\left(q-r-m+2,\tfrac{m(1-\alpha-\beta)+\beta)}{m\alpha},\tfrac{\gamma}{\alpha \overline \gamma}\right).
\end{align}
for $m\in\mathbb{Z}^+$.
\end{lemma}

\begin{proof}
See Appendix \ref{ap2}.
\end{proof}

\begin{lemma}\label{lemma3}
Let $\gamma\sim\mathcal{F}_{\rm SOSF}\left(\alpha,\beta,m;\overline\gamma\right)$. Then, for $m\in\mathbb{Z}^+$ the \ac{GMGF} of $\gamma$ is given by
\begin{align}
\label{MGF SOSFS_1}
\mathcal{M}_{\gamma}^{(n)}(s)&=\sum_{q=0}^{n}\tbinom{n}{q}\frac{(-1)^{q+1}(m-n-q)_{n-q}(m)_q}{s^{n+1}\bar\gamma \alpha} \nonumber \times \\
&\sum_{i=0}^{n-q}\sum_{j=0}^{q}\sum_{r=0}^{m-1-n+q}\tbinom{n-q}{i}\tbinom{q}{j}\tbinom{m-1-n+q}{r} \nonumber  \times \\
&c^{n-q-i}d^{q-j}a(s)^{m-1-n+q-r} \Gamma(1+r+i+j)\nonumber \times \\
&{\rm U}(m+q,m+q-r-i-j,b(s)).
\end{align}
\begin{figure*}[ht!]
\begin{align}
\label{MGF SOSFS_2}
\mathcal{M}_{\gamma}^{(n)}(s)=&\sum_{q=0}^{n-m}\tbinom{n}{q}\tfrac{(-1)^{q+1}(m-n-q)_{n-q}(m)_q}{s^{n+1}\bar\gamma \alpha}\sum_{i=0}^{n-q}\sum_{j=0}^{q}\tbinom{n-q}{i}\tbinom{q}{j} c^{n-q-i}d^{q-j}\times \nonumber\\
&\left[\sum_{k=1}^{n+1-m-q}A_k(s) {\rm U}(k,k,a(s))+ \sum_{k'=1}^{m+q}B_k(s) {\rm U}(k',k',b(s))\right] + \nonumber\\
\vspace{2mm}&\sum_{q=n+1-m}^{n}\tbinom{n}{q}\tfrac{(-1)^{q+1}(m-n-q)_{n-q}(m)_q}{s^{n+1}\bar\gamma \alpha}\sum_{i=0}^{n-q}\sum_{j=0}^{q}\tbinom{n-q}{i}\tbinom{q}{j}c^{n-q-i}d^{q-j} \times \nonumber\\
&\sum_{r=0}^{m-1-n+q}\tbinom{m-1-n+q}{r} a(s)^{m-1-n+q-r} \Gamma(1+r+i+j){\rm U}(m+q,m+q-r-i-j,b(s)).
\end{align}
\hrulefill
\end{figure*}
for $m \geq n+1$, and in \eqref{MGF SOSFS_2} at the top of next page for $m < n+1$, where $A_k(s)$ and $B_k(s)$ are the partial fraction expansion coefficients given by
\begin{align}
\label{A_k}
A_k(s)=&\sum_{l=0}^{\sigma_1-k} \tfrac{\tbinom{\sigma_1-k}{l}}{(\sigma_1-k)!}(i+j-\sigma_1+k+l+1)_{\sigma_1-k-l} \times \nonumber \\
&(\sigma_2)_l (-1)^l(-a)^{i+j-\sigma_1+k+l}(b-a)^{-\sigma_2-l}, \nonumber \\
B_k(s)=&\sum_{l=0}^{\sigma_2-k}\tfrac{\tbinom{\sigma_2-k}{l}}{(\sigma_2-k)!}(i+j-\sigma_2+k+l+1)_{\sigma_2-k-l}\times \nonumber \\
&(\sigma_1)_l (-1)^l(-b)^{i+j-\sigma_2+k+l}(a-b)^{-\sigma_1-l},
\end{align}
with $\sigma_1=n+1-m-q$ and $\sigma_2=m+q$, and where ${\rm U}\left(\cdot,\cdot,\cdot \right)$ is Tricomi's confluent hypergeometric function \cite[(13.1)]{NIST}.
\end{lemma}

\begin{proof}
See Appendix \ref{ap3}.
\end{proof}

\begin{lemma}\label{lemma4}
Let $\gamma\sim\mathcal{F}_{\rm SOSF}\left(\alpha,\beta,m;\overline\gamma\right)$. Then, for $m\in\mathbb{Z}^+$ the $n^{\rm th}$ moment of $\gamma$ is given by
\begin{align}
\label{Moments}
\mathbb{E}[\gamma^n]=&(\bar\gamma \alpha)^n \sum_{q=0}^{n}\tbinom{n}{q}(-1)^{q-n}(m-n+q)_{n-q}(m)_q \nonumber \times \\
&\sum_{i=0}^{n-q}\sum_{j=0}^{q}\tbinom{n-q}{i}\tbinom{q}{j}c^{n-q-i}d^{q-j} (i+j)!
\end{align}
\end{lemma}
\begin{proof}
See Appendix \ref{ap4}
\end{proof}

\newcommand\figureSize{1}
\section{Numerical results}
\label{Sec:4}

\begin{figure}[t]
    \centering
    \includegraphics[width=\figureSize \columnwidth]{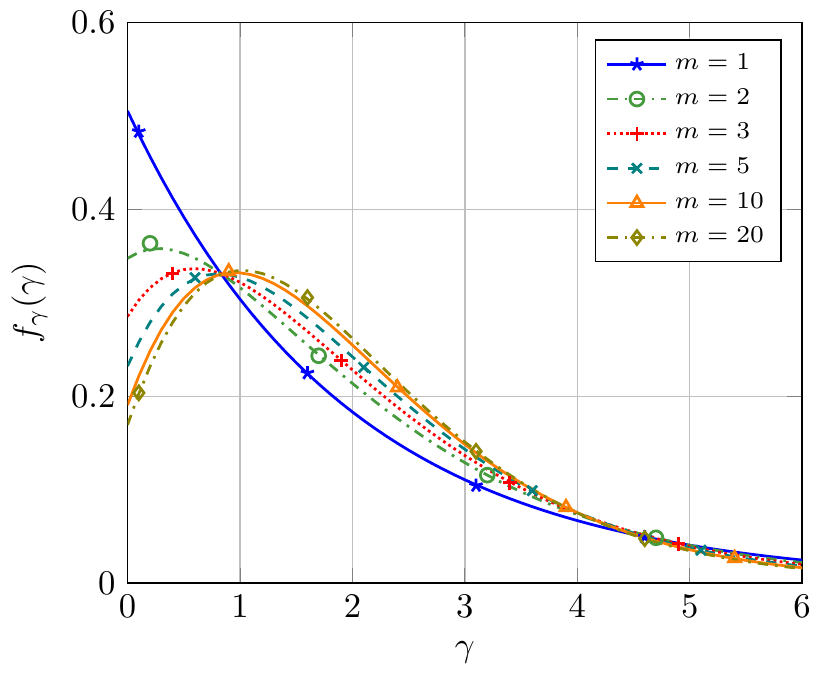}
    \caption{\ac{PDF} of \ac{fSOSF} model for different values of $m$. Parameter values are $\alpha=0.1$, $\beta=0.7$ and $\overline{\gamma}_{\rm dB}=3$dB. Theoretical values (\ref{Eq_PDF}) are represented with lines. Markers correspond to \ac{MC} simulations.}
    \label{fig:pdf}
\end{figure}

\begin{figure}[t]
    \centering
    \includegraphics[width=\figureSize \columnwidth]{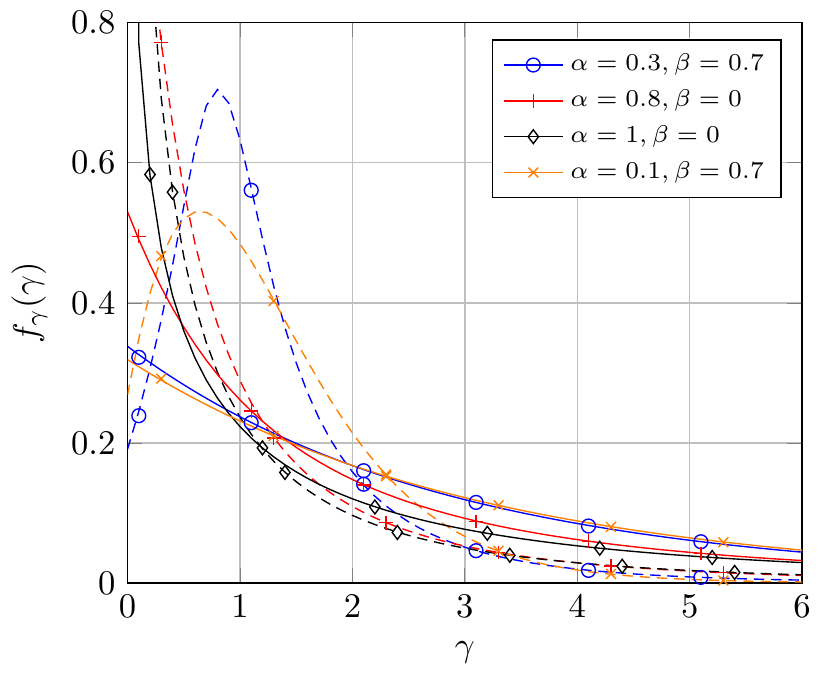}
    \caption{\ac{PDF} comparison for different values of $\alpha$ and $\beta$. Solid/dashed lines obtained with (\ref{Eq_PDF}) correspond to ($m=1$, $\overline{\gamma}_{\rm dB}=5$dB) and ($m=20$, $\overline{\gamma}_{\rm dB}=1$dB), respectively. Markers correspond to \ac{MC} simulations.}
    \label{fig:pdf2}
\end{figure}

{In Fig. \ref{fig:pdf}, we represent the \ac{PDF} of the \ac{fSOSF} fading model in Lemma \ref{lemma1}, for different values of the \ac{LoS} fluctuation severity parameter $m$. Parameter values are $\alpha=0.1$, $\beta=0.7$ and $\overline{\gamma}_{\rm dB}=3$dB. \ac{MC} simulations are also included as a sanity check. See that, as we increase the fading severity of the \ac{LoS} component (i.e., $\downarrow m$), the probability of occurrence of low SNR values increases, as well as the variance of the distribution.}

{In Fig. \ref{fig:pdf2} we analyze the impact of $\alpha$ and $\beta$ on the \ac{PDF} of the fSOSF model. Two scenarios have been considered: one with low fading severity and low average \ac{SNR} ($m=20$, $\overline{\gamma}_{\rm dB}=1 dB$), and another with higher fading severity and higher average \ac{SNR} ($m=1$, $\overline{\gamma}_{\rm dB}=5 dB$).} { In the case of mild fluctuations of the \ac{LoS} component, the effect of $\beta$ dominates to determine the shape of the distribution, observing a bell-shaped \ac{PDF}s with higher $\beta$ values.  Conversely, the value of $\alpha$ becomes more influential for the left tail of the distribution. We see that lower values of alpha make lower SNR values more likely, which implies an overall larger fading severity.}

{Finally, we analyze the \ac{OP} under \ac{fSOSF} which is defined as the probability that the instantaneous SNR takes a value below a given threshold, $\gamma_{\rm th}$. It can be obtained from the \ac{CDF} (\ref{CDF_promediado}) as }
\begin{equation}
    \label{Eq_OP}
    {\rm OP}=F_\gamma(\gamma_{\rm th}).
\end{equation}
{Fig. \ref{fig:op} shows the \ac{OP} under \ac{fSOSF} model, for different values of the parameter $m$. Additional parameters are set to $\alpha=0.1$ and $\beta=0.7$, and two threshold values are considered: $\gamma_{\rm th}=3$dB and $\gamma_{\rm th}=1$dB. We see that a $2$dB change in the thresold SNR is translated into a $\sim5$dB power offset in terms of \ac{OP} performance. We observe that as the severity of fading increases (i.e., $\downarrow m$), the less likely it is to exceed the threshold value $\gamma_{\rm th}$, i.e., the higher the \ac{OP}. In all instances, the diversity order (i.e., the down-slope decay of the OP) is one, and the asymptotic OP in \eqref{Eq_OP_asymp} tightly approximates the exact OP, which is given by} 
\begin{align}
    \label{Eq_OP_asymp}
		{\rm OP} & \left(\alpha,\beta,m;\overline\gamma,\gamma_{\rm th} \right)\approx\frac{\gamma_{\rm th}}{\alpha \overline\gamma } \sum_{j=0}^{m-1}  \tbinom {m-1}{j} \left(\tfrac{1-\beta-\alpha}{\alpha}\right)^{m-1-j} \times \nonumber\\
		&\Gamma(1+j) {\rm U}(m,m-j,\tfrac{1}{\alpha}-\tfrac{\beta}{\alpha}\left(\tfrac{m-1}{m} \right)-1).
\end{align}
Expression \eqref{Eq_OP_asymp} can be readily derived by integration over the asymptotic \ac{OP} of the underlying Rician shadowed model \cite[eq. (22)]{Lopez2022}.
\begin{figure}[t]
\centering
\includegraphics[width=\figureSize \columnwidth]{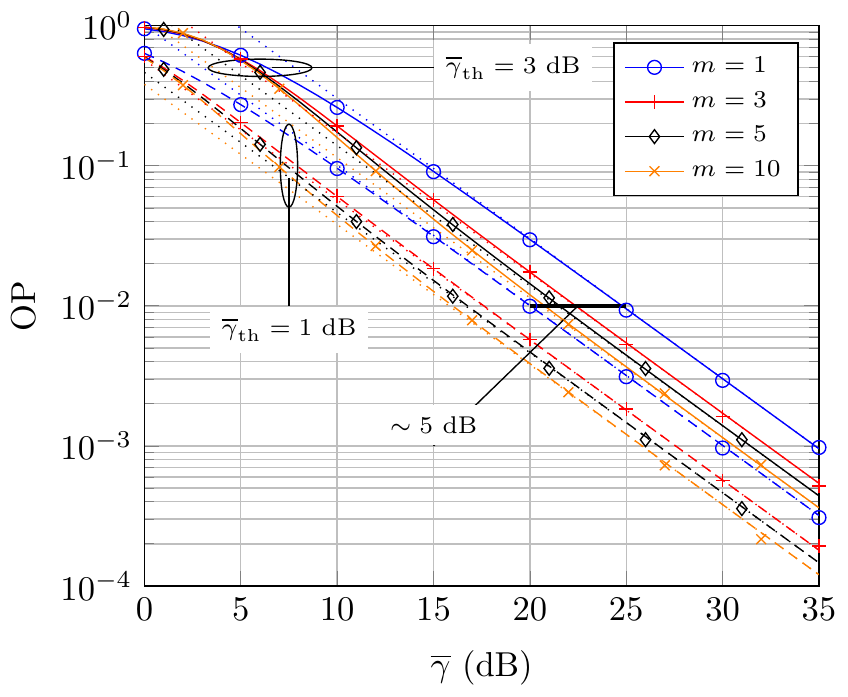}
\caption{\ac{OP} as a function of $\overline\gamma$, for different values of $m$. Parameter values are $\alpha=0.1$ and $\beta=0.7$. Solid/dashed lines correspond to $\gamma_{\rm th}=3$ and $\gamma_{\rm th}=1$ respectively. Theoretical values (\ref{Eq_OP}) are represented with lines. Markers correspond to \ac{MC} simulations.}
\label{fig:op}
\end{figure}

\section{Conclusions}
We presented a generalization of Andersen's \ac{SOSF} model by incorporating random fluctuations on its dominant specular component, yet without incurring in additional complexity. We provided closed-form expressions for its probability and cumulative distribution functions, as well as for its generalized Laplace-domain statistics and raw moments. Some insights have been provided on how the set of parameters ($\alpha$, $\beta$ and $m$) affect propagation, and its application to performance analysis has been exemplified through an outage probability analysis.
\appendices

\section{ Proof of Lemma~\ref{lemma1}}
\label{ap1}
Noting that $x=|G_3|^2$ is exponentially distributed with unitary mean, we can compute the distribution of $\gamma$ by averaging over all possible values of $x$ as:
\begin{equation}
f_{\gamma}(\gamma)=\int_0^\infty f_{\gamma_x}(\gamma;x) e^{-x}dx.
\label{Eq_int_PDF}
\end{equation}
The \ac{PDF} of $\gamma_x$ is that of a squared Rician shadowed \ac{RV}, which for integer $m$ is given by \cite[eq. (5)]{Martinez2017}
\begin{equation}
f_{\gamma_x}(\gamma;x)=\sum_{j=0}^{m-1}B_j \left( \tfrac{m-j}{\omega_B}\right)^{m-j} \tfrac{\gamma^{m-j-1}}{(m-j-1)!} e^{-\frac{\gamma(m-j)}{\omega_B}},
\label{Eq_PDF_RS}
\end{equation}
where
\begin{equation}
B_j=\binom{m-1}{j} \left( \tfrac{m}{K_x+m}\right)^j \left( \tfrac{K_x}{K_x+m}\right)^{m-j-1},
\label{Eq_Bj}
\end{equation}
\begin{equation}
\omega_B=(m-j)\left( \tfrac{K_x}{K_x+m}\right)\left( \tfrac{\overline{\gamma}_x}{1+K_x}\right).
\label{Eq_wB}
\end{equation}
with $K_x$ and $\bar{\gamma}_x$ given in \eqref{Eq_gamma_x} and \eqref{Eq_K_x}.
\color{black} Substituting (\ref{Eq_PDF_RS}) into (\ref{Eq_int_PDF}), using the change of variables $t=\tfrac{1}{\alpha}\left(1-\beta(\tfrac{m-1}{m})-\alpha(1-x)\right)$ and taking into account that ${(t-\tfrac{\beta}{\alpha m})^j=\sum_{r=0}^{j}\binom{j}{r}t^j (\tfrac{\beta}{\alpha m})^{j-r}}$, the final expression for the PDF is derived.

\section{ Proof of Lemma~\ref{lemma2}}
\label{ap2}
The \ac{CDF} of the  \ac{fSOSF} model can also be obtained by averaging the \ac{CDF} of $\gamma_{x}$, i.e., the Rician shadowed \ac{CDF} over the exponential distribution:
\begin{equation}
\label{CDF_promediado}
F_{\gamma}(\gamma)=\int_0^\infty F_{\gamma_x}(\gamma;x) e^{-x}dx.
\end{equation}
\color {black}
For the case of integer $m$, a closed-form expression for the Rician shadowed CDF is presented in \cite[eq. (10)]{Martinez2017}, i.e. 
\begin{equation}
\label{CDF_Rice_Shadowed}
F_{\gamma_x}(\gamma;x)=1-\sum_{j=0}^{m-1}B_je^{\frac{-\gamma (m-j)}{\omega_B}} \sum_{r=0}^{m-j-1}\tfrac{1}{r!}\left( \tfrac{\gamma (m-j)}{\omega_B}\right)^r,
\end{equation}

Substituting \eqref{CDF_Rice_Shadowed} in \eqref{CDF_promediado} and following the same approach used in the previous appendix, we obtain the final expression.

\section{ Proof of Lemma~\ref{lemma3}}
\label{ap3}
Following the same procedure, the generalized \ac{MGF} of the \ac{fSOSF} model denoted as $\mathcal{M}_{\gamma}^{(n)}(s)$ can be obtained by averaging the generalized \ac{MGF} of $\gamma_{x}$, i.e., the Rician shadowed generalized \ac{MGF} over the exponential distribution:
\begin{equation}
\label{MGF_promediado}
\textcolor[rgb]{0,0,0}{\mathcal{M}_{\gamma}^{(n)}(s)=\int_0^\infty \mathcal{M}_{\gamma_x}^{(n)}(s;x) e^{-x}dx.}
\end{equation}
\textcolor[rgb]{0,0,0}{A closed-form expression for $M_{\gamma_x}(s;x)$ for integer $m$ is provided in \cite[eq. (26)]{Martinez2017}}
\begin{equation}
\label{MGF_Rician Shadowed}
M_{\gamma_x}(s;x)=\frac{m^m (1+K_x)}{\overline \gamma_x(K_x+m)^m}\frac{\left(s-\frac{1+K_x}{\overline \gamma_x} \right)^{m-1}}{\left(s-\frac{1+K_x}{\overline \gamma_x}\frac{m}{K_x+m} \right)^{m}}
\end{equation}

Substituting \eqref{Eq_gamma_x} and \eqref{Eq_K_x} into \eqref{MGF_Rician Shadowed} the expression for $M_{\gamma_x}(s;x)$ can be rewritten as
\begin{equation}
\label{MGF_Rician Shadowed__F1_F2}
M_{\gamma_x}(s;x)=F_1(s;x)\cdot F_2(s;x),
\end{equation}
where
\begin{align}
\label{MGF_Rician Shadowed_F1_F2_bis}
F_1(s;x)=&-\left[s\bar{\gamma}\left(1-\beta-\alpha(1-x)\right)-1\right]^{m-1},\\
F_2(s;x)=&\left[s\bar{\gamma}\left(1-\beta\left(\tfrac{m-1}{m} \right)-\alpha(1-x)\right)-1\right]^{-m}.
\end{align}

Next, we compute the $n$-th derivative of \eqref{MGF_Rician Shadowed__F1_F2} wich yields the Rician shadowed generalized \ac{MGF}
\begin{equation}
\label{MGF_Rician Shadowed_derivative}
\mathcal{M}_{\gamma_x}^{(n)}(s;x)=\tfrac{\partial^n M_{\gamma_x}(s;x)}{\partial s^n}=\sum_{q=0}^{n}\tbinom{n}{q}F_1^{(n-q)}(s;x)\cdot F_2^{(q)}(s;x),
\end{equation}
where $F_i^{(n)}(s;x)$ denotes the $n$-th derivative of $F_i(s;x)$ with respect to $s$.

\begin{align}
\label{F1_derivative}
F_1^{(n-q)}(s;x)=&-(m-n+q)_{n-q} s^{m-1-n+q} (\bar\gamma\alpha)^{m-1} \times \nonumber \\
&(x+a(s))^{m-1-n+q}(x+c)^{n-q}
\end{align}
\begin{align}
\label{F2_derivative}
F_2^{(q)}(s;x)=&(-1)^q (m)_{q} s^{-m-q} (\bar\gamma\alpha)^{-m} \times \nonumber\\
&(x+b(s))^{-m-q}(x+d)^{q},
\end{align}
where $(z)_j$ denotes the Pochhammer symbol and where
\begin{align}
\label{a}
a(s)=&\tfrac{s\bar\gamma(1-\beta-\alpha)-1}{s\bar\gamma\alpha},\\
b(s)=&\tfrac{s\bar\gamma\left(1-\beta\left( \tfrac{m-1}{m}\right)-\alpha\right)-1}{s\bar\gamma\alpha},\\
c=&\tfrac{1-\beta}{\alpha}-1,\\
d=&\tfrac{1-\beta\left( \tfrac{m-1}{m}\right)}{\alpha}-1.
\end{align}

From \eqref{MGF_Rician Shadowed_derivative}, \eqref{F1_derivative} and \eqref{F2_derivative} notice that $\mathcal{M}_{\gamma_x}^{(n)}(s;x)$ is a rational function of $x$ with two real positive zeros at $x=c$ and $x=d$, one real positive pole at $x=b(s)$ and one real positive zero or pole at $x=a(s)$ depending on whether $m\geq n+1$ or not. Integration of \eqref{MGF_promediado} is feasible with the help of \cite[eq. 13.4.4]{NIST}
\begin{equation}
\label{Integral Tricomi Wolfram}
\int_{0}^{\infty}\frac{x^i}{(x+p)^j}e^{-x}dx=\Gamma(i+1){\rm U}(j,j-i,p).
\end{equation}
where ${\rm U}\left(\cdot,\cdot,\cdot \right)$ is Tricomi's confluent hypergeometric function \cite[(13.1)]{NIST}. We need to expand $\mathcal{M}_{\gamma_x}^{(n)}(s;x)$ in partial fractions of the form $\tfrac{x^i}{(x+p)^j}$. Two cases must be considered:

\begin{itemize}
\item $m \geq n+1$

In this case there is only one pole at $x=b(s)$ and no partial fraction expansion is required. Using the fact that $(x+p)^n=\sum_{i=0}^{n}x^ip^{n-i}$, then \eqref{MGF SOSFS_1} is obtained. 

\item $m < n+1$

Now, there are two poles ($x=a(s)$, $x=b(s)$). After performing partial fraction expansion we obtain \eqref{MGF SOSFS_2}.
\end{itemize}

\section{ Proof of Lemma~\ref{lemma4}}
\label{ap4}
Using the definition of $\gamma_x$, we can write
\begin{equation}
\label{Momentos_promediado}
\mathbb{E}[\gamma^n]=\int_0^\infty \mathbb{E}[\gamma_x^n] e^{-x}dx.
\end{equation}
where $\mathbb{E}[\gamma_x^n]=\lim_{s \rightarrow 0^{-}} M_{\gamma_x}^{(n)}(s;x)$ where $M_{\gamma_x}^{(n)}(s;x)$ is given in \eqref{MGF_Rician Shadowed_derivative}. Performing the limit and using the integral $\int_{0}^{\infty}x^pe^{-x}dx=p!$, the proof is complete.

\bibliographystyle{ieeetr}
\bibliography{SOSF2}
\end{document}